\documentclass[12pt,reqno]{article}

\usepackage[usenames]{color}
\usepackage{amssymb}
\usepackage{graphicx}
\usepackage{amscd}
\usepackage[colorlinks=true,
linkcolor=webgreen,
filecolor=webbrown,
citecolor=webgreen]{hyperref}

\definecolor{webgreen}{rgb}{0,.5,0}
\definecolor{webbrown}{rgb}{.6,0,0}

\usepackage{color}
\usepackage{fullpage}
\usepackage{float}
\usepackage{cite}

\usepackage{graphics,amsmath,amssymb}
\usepackage{amsthm}
\usepackage{amsfonts}
\usepackage{latexsym}
\usepackage{epsf}
\usepackage{epsfig}
\usepackage{epstopdf}

\newcommand{\seqnum}[1]{\href{http://oeis.org/#1}{\underline{#1}}}

\theoremstyle{plain}
\newtheorem{theorem}{Theorem}
\newtheorem{corollary}[theorem]{Corollary}
\newtheorem{lemma}[theorem]{Lemma}

\newtheorem{prop}[theorem]{Proposition}

\theoremstyle{definition}

\newtheorem{example}[theorem]{Example}
\newtheorem{examples}[theorem]{Examples}

\theoremstyle{remark}

\begin{document}


\begin{center}
\vskip 1cm{\LARGE\bf Using Graph Theory\\ 
\vskip .04in
to Derive Inequalities for the Bell Numbers}
\vskip 1cm
\normalsize
Alain Hertz\\
Department of Mathematics and Industrial Engineering\\
Polytechnique Montr\'eal and GERAD\\
Montr\'eal, Qu\'ebec H3T 1J4 \\
Canada\\
\href{mailto:alain.hertz@gerad.ca}{\tt alain.hertz@gerad.ca}\\
\ \\
Anaelle Hertz\\
Department of Physics\\
University of Toronto\\
Toronto, Ontario M5S 1A7\\
Canada\\
\href{mailto:ahertz@physics.utoronto.ca}{\tt ahertz@physics.utoronto.ca}\\
\ \\
Hadrien M\'elot \\
Computer Science Department - Algorithms Lab\\
University of Mons\\
 7000 Mons \\
Belgium\\
\href{mailto:hadrien.melot@umons.ac.be}{\tt hadrien.melot@umons.ac.be}\\

\end{center}

\vskip .2 in

\newcommand {\stirlings}[2]{\genfrac\{\}{0pt}{}{#1}{#2}}
\newcommand{\B}{\mathcal{B}}
\newcommand{\T}{\mathcal{T}}
\newcommand{\A}{\mathcal{A}}

\begin{abstract}
The Bell numbers count the number of different ways to partition a set of $n$ elements while the graphical Bell numbers count the number of non-equivalent partitions of the vertex set of a graph into stable sets. This relation between graph theory and integer sequences has motivated us to study properties on the average number of colors in the non-equivalent colorings of a graph to discover new nontrivial inequalities for the Bell numbers. Examples are given to illustrate our approach.
\end{abstract}

\section{Introduction}

The Bell numbers $(B_n)_{n\geq 0}$ count the number of different ways to partition a set that has exactly $n$ elements. Starting with $B_0 = B_1 = 1$, the first few Bell numbers are 1, 1, 2, 5, 15, 52, 203 (sequence \seqnum{A141390}). The integer $B_n$ can be defined as the sum 
$$B_n=\sum_{k=0}^{n}\stirlings{n}{k}$$
where $\stirlings{n}{k}$ is the Stirling number of the second kind, with parameters $n$ and $k$ (i.e., the number of partitions of a set of $n$ elements into $k$ blocks). Dobi\'nski's formula \cite{Dob} gives
$$B_n=\frac{1}{e}\sum_{k=0}^{\infty}\frac{k^{n}}{k!}.$$

The 2-Bell numbers $(T_n)_{n\geq 0}$ count the total number of blocks in all partitions of a set of $n+1$ elements. Starting with $T_0 = 1$ and $T_1=3$, the first few 2-Bell numbers are 1, 3, 10, 37, 151, 674 (sequence \seqnum{A005493}). More formally, the integer $T_n$ is defined as
$$T_{n}=\sum_{k=0}^{n+1}k\stirlings{n+1}{k}=B_{n+2}-B_{n+1}.$$

Odlyzko and Richmond \cite{OR} have studied the average number $A_n$ of blocks in a partition of a set of $n$ elements, which can be defined as 
$$A_n=\frac{T_{n-1}}{B_n}.$$

A concept very close to the Bell numbers is also defined in graph theory. More precisely, a coloring of a graph $G$ is an assignment of colors to its vertices such that adjacent vertices have different colors. The chromatic number $\chi(G)$ is the minimum number of colors in a coloring of $G$. Two colorings are equivalent if they induce the same partition of the vertex set into color classes. For an integer $k>0$, we define $S(G,k)$ as the number of proper non-equivalent colorings of a graph $G$ that use exactly $k$ colors. Since $S(G,k)=0$ for $k< \chi(G)$ or $k>n$, the total number $\B(G)$ of non-equivalent colorings of a graph $G$ is defined as 
$$\B(G)=\sum_{k=0}^nS(G,k)=\sum_{k=\chi(G)}^nS(G,k).$$

In other words, $\B(G)$ is the number of partitions of the vertex set of $G$ whose blocks are stable
sets (i.e., sets of pairwise non-adjacent vertices). This invariant has been studied by several authors in the last few years \cite{absil,Dun10,Duncan09,GT13,numcol,KN14} under the name of (graphical) Bell number
of $G$.

Let $\T(G)$ be the total number of stable sets in the set of non-equivalent colorings of a graph $G$. More precisely, we define 
$$\T(G)=\sum_{k=\chi(G)}^nkS(G,k).$$

We are interested in computing the average number $\A(G)$ of colors in the non-equivalent  colorings of $G$, that is
$$\A(G)=\frac{\T(G)}{\B(G)}.$$

\noindent Clearly, 
$\B(\overline{K}_n)=B_n, \T(\overline{K}_n)=T_{n-1}=B_{n+1}-B_{n} \mbox{, and } \A(\overline{K}_n)=\frac{B_{n+1}-B_{n}}{B_n}$
where $\overline{K}_n$ is the empty graph with $n$ vertices. As another example, consider the cycle $C_5$ on 5 vertices. As shown in Figure \ref{fig:C5}, there are five colorings of $C_5$ with 3 colors, five with 4 colors, and one with 5 colors, which gives $\B(C_5)=11$, $\T(C_5)=40$ and $\A(G)=\frac{40}{11}.$

\begin{figure}[!hbtp]
	\centering
	\includegraphics[scale = 0.9]{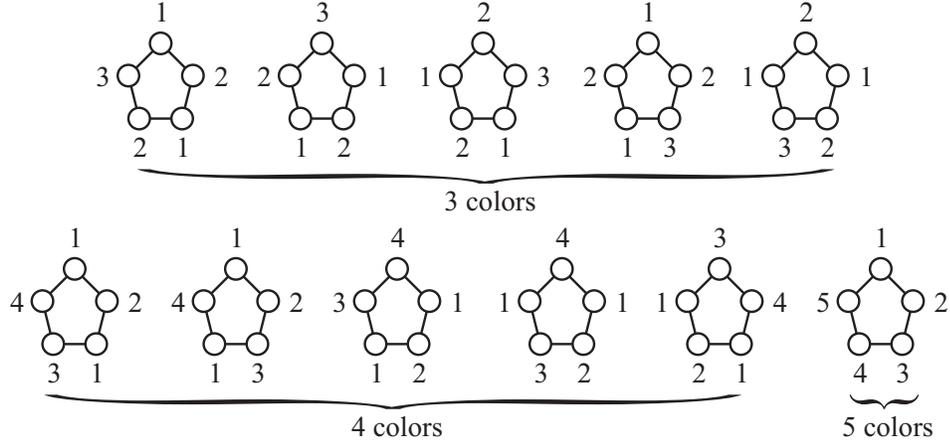}
	\caption{The non-equivalent colorings of $C_{5}$.}\label{fig:C5}
\end{figure}

This close link between Bell numbers and graph colorings  indicates that it is possible to use graph theory to derive inequalities for the Bell numbers. This is the aim of this article. The next section gives values of $\A(G)$ for some families of graphs and basic properties involving $\A(G)$. In Section \ref{sec:3}, we give several examples of inequalities for the Bell numbers that can be deduced from relations involving $\A(G)$. 

Let $u$ and $v$ be two vertices in a graph $G$. We denote by $G_{\mid uv}$ the graph obtained by identifying (merging) the vertices $u$ and $v$ and, if $u$ and $v$ are adjacent vertices, by removing the edge that links $u$ and $v$. If parallel edges are created, we keep only one. Also, if $u$ is adjacent to $v$, we denote by $G - uv$  the graph obtained from $G$ by removing the edge that links $u$ with $v$, while if $u$ is not adjacent to $v$, we denote by $G+uv$ the graph obtained by linking $u$ with $v$. In what follows, we let $K_n$, $P_n$ and $C_n$ be the complete graph of order $n$, the path of order $n$, and the cycle of order $n$, respectively. We denote the disjoint union of two graphs $G_1$ and $G_2$ by $G_1\cup G_2$. We refer to Diestel \cite{Diestel} for basic notions of graph theory that are not defined here.

\section{Some values and properties of $\A(G)$}\label{sec:2}

The \emph{deletion-contraction} rule (also often called the {\it Fundamental Reduction Theorem} \cite{DKT05}) is a well known method to compute $\B(G)$ \cite{Duncan09,KN14}.  More precisely, let $u$ and $v$ be any pair of distinct vertices of $G$. We have,
\begin{align} \label{rec_minusk}S(G, k) = S(G - uv, k) - S(G_{\mid uv}, k)& \mbox{ for every pair }u,v\mbox{ of adjacent vertices in }G,\\
S(G, k) = S(G + uv, k) + S(G_{\mid uv}, k) &\mbox{ for every pair }u,v\mbox{ of non-adjacent vertices in }G.
\end{align}
It follows that
\begin{align}
\left.
\begin{array}{ll}
\B(G) = \B(G - uv) - \B(G_{\mid uv})\\
\T(G) = \T(G - uv) - \T(G_{\mid uv})
\end{array}
\right\}& \mbox{ for every pair }u,v\mbox{ of adjacent vertices in }G,\label{rec_minus}\\
\left.
\begin{array}{ll}
\B(G) = \B(G + uv) + \B(G_{\mid uv})\\
\T(G) = \T(G + uv) + \T(G_{\mid uv})
\end{array}
\right\}& \mbox{ for every pair }u,v\mbox{ of non-adjacent vertices in }G.\label{rec_plus}
\end{align}

Let $v$ be a vertex in a graph $G$. We denote by $G-v$ the graph obtained from $G$ by removing $v$ and all its incident edges. A vertex of a graph $G$ is {\it dominating} if it is adjacent to all other vertices of $G$, and it is {\it simplicial} if its neighbors are pairwise adjacent. 

\begin{prop}\label{prop:1} If $G$ has a dominating vertex $v$, then $\A(G)=1+\A(G-v)$.
\end{prop}
\begin{proof}
	Clearly, $S(G,k)=S(G-v,k-1)$ for all $k$, which implies
\begin{align*}
\B(G)&=\sum_{k=\chi(G)}^nS(G,k)\\
 &=\sum_{k=\chi(G)}^nS(G-v,k-1)=\sum_{k=\chi(G-v)}^{n-1}S(G-v,k)= \B(G-v)\\
 \mbox{and}\quad\quad\\
\T(G)&=\sum_{k=\chi(G)}^nkS(G,k)\\
 &= \sum_{k=\chi(G)}^nkS(G-v,k-1)=
\sum_{k=\chi(G-v)}^{n-1}(k+1)S(G-v,k)
= \T(G-v)+\B(G-v).
\end{align*}
Hence, $\A(G)= \displaystyle\frac{ \T(G-v)+\B(G-v)}{\B(G-v)}=1+\frac{ \T(G-v)}{\B(G-v)}=1+\A(G-v).$
\end{proof}

Duncan \cite{Duncan09} has proved that if $G$ is a tree, then $S(G,k)=\stirlings{n-1}{k}$ for all $k=1,\cdots,n$. This leads to our second Proposition.

\begin{prop}\label{prop:2} Let $G$ be a tree of order $n$. Then $\B(G)=B_{n-1}$ and $\T(G)=B_n$.\end{prop}
\begin{proof}
	Since $S(G,k)=\stirlings{n-1}{k-1}$, we immediately get 	\begin{align*}
	\B(G)=&\sum_{k=1}^{n}\stirlings{n-1}{k-1}=\sum_{k=0}^{n-1}\stirlings{n-1}{k}=B_{n-1}\\
	 \mbox{and}\hspace{2cm}\\
	 \T(G)=&\sum_{k=1}^{n}k\stirlings{n-1}{k-1}=\sum_{k=0}^{n-1}(k+1)\stirlings{n-1}{k}\\
	=&\sum_{k=0}^{n-1}k\stirlings{n-1}{k}+\sum_{k=0}^{n-1}\stirlings{n-1}{k}=(B_{n}-B_{n-1})+B_{n-1}=B_n.
	\end{align*}
\end{proof}

\begin{prop}\label{prop:3}  Let $T\cup pK_1$ be the graph obtained from a tree $T$ of order $n \ge 1$ by adding $p$ isolated verices. Then
	$\displaystyle\B(T\cup pK_1)=\sum_{i = 0}^{p} {p \choose i} B_{n+i-1}$ and 
	$\displaystyle\T(T\cup pK_1) = \sum_{i = 0}^{p} {p \choose i} B_{n+i}$.
\end{prop}
\begin{proof}
	For $p=0$, the result follows from Proposition \ref{prop:2}. For larger values of $p$, we proceed by induction. Let $T'$ be the tree obtained from $T$ by adding a new vertex and linking it to
	exactly one vertex in $T$.  Equations~(\ref{rec_plus}) give :
	\begin{align*}
	\B(T\cup pK_1)&=\B(T'\cup (p-1)K_1)+\B(T\cup (p-1)K_1)\\
	&=\sum_{i = 0}^{p-1} {p-1 \choose i}B_{n+i}+\sum_{i = 0}^{p-1} {p-1 \choose i}B_{n+i-1}\\
	&=\sum_{i = 1}^{p} {p-1 \choose i-1}B_{n+i-1}+\sum_{i = 0}^{p-1} {p-1 \choose i}B_{n+i-1}\\
	&=B_{n+p-1}+\sum_{i = 1}^{p-1} \left({p-1 \choose i-1}+{p-1 \choose i}\right)B_{n+i-1}+B_{n-1}\\
	&=\sum_{i = 0}^{p} {p \choose i} B_{n+i-1}.
	\end{align*}
	The proof for $\T(T\cup pK_1)$ is similar.
\end{proof}

\begin{prop}\label{prop:4} Let $C_n$ be a cycle of order $n \ge 3$. Then, 
	$$\displaystyle\B(C_n)=\sum_{j = 1}^{n-1} (-1)^{j+1} B_{n-j}\quad\mbox{ and }\quad\displaystyle\T(C_n)=\sum_{j = 1}^{n-1} (-1)^{j+1} B_{n-j+1}.$$
\end{prop}
\begin{proof}
	Duncan \cite{Duncan09} proved that $\B(C_n) = \sum_{j = 1}^{n-1} (-1)^{j+1} B_{n-j}$. It is therefore sufficient to prove that $\T(C_n) = \sum_{j = 1}^{n-1} (-1)^{j+1} B_{n-j+1}$.
	
	If $n = 3$, then $\T(C_3) = 3 = B_{3} - B_{2}$.  If $n > 3$, Equations~(\ref{rec_minus}) together with the fact that $P_n$ is a tree give $\T(C_n) = \T(P_n) - \T(C_{n-1}) = B_{n} - \T(C_{n-1})$, and the result follows by induction.
\end{proof}

\begin{prop}\label{prop:5}  Let $C_n\cup pK_1$ be the graph obtained from a cycle of order $n \ge 3$ by adding $p$ isolated verices.
	Then
	$$\B(C_n\cup pK_1)=\sum_{j=1}^{n-1}(-1)^{j+1}\sum_{i = 0}^{p} {p \choose i} B_{n+i-j}\;\mbox{ and }\;\T(C_n\cup pK_1) =\sum_{j=1}^{n-1}(-1)^{j+1}\sum_{i = 0}^{p} {p \choose i} B_{n+i-j+1}.$$
\end{prop}
\begin{proof}
	For $p=0$, the result follows from Proposition \ref{prop:4}. For larger values of $p$, we proceed by induction. If $n=3$ then Equations~(\ref{rec_minus}) and Proposition \ref{prop:3} give
	\begin{align*}
	\B(C_3\cup pK_1)&=\B(P_{3}\cup pK_1)-\B(P_{2}\cup pK_1) \\
	&= \sum_{i = 0}^{p} {p \choose i} B_{3+i-1}-\sum_{i = 0}^{p} {p \choose i} B_{2+i-1}\\
	&=\sum_{j=1}^{2}(-1)^{j+1}\sum_{i = 0}^{p} {p \choose i} B_{3+i-j}.
	\end{align*}
	Hence, the result is valid for $n=3$. So assume $n>3$ and that the statement holds for smaller values of $n$:
	\begin{align*}
	\B(C_n\cup pK_1)&=\B(P_{n}\cup pK_1)-\B(C_{n-1}\cup pK_1)\\
	&=\sum_{i = 0}^{p} {p \choose i}B_{n+i-1}-\sum_{j=1}^{n-2}(-1)^{j+1}\sum_{i = 0}^{p} {p \choose i}B_{n+i-j-1}\\
	&=\sum_{i = 0}^{p} {p \choose i}B_{n+i-1}+\sum_{j=2}^{n-1}(-1)^{j+1}\sum_{i = 0}^{p} {p \choose i}B_{n+i-j}\\
	&=\sum_{j=1}^{n-1}(-1)^{j+1}\sum_{i = 0}^{p} {p \choose i} B_{n+i-j}.
	\end{align*}
	The proof for $\T(C_n\cup pK_1)$ is similar.
\end{proof}

\begin{prop}\label{prop:6}  Let $G$ be a graph with a simplicial vertex $v$. Then $\A(G)>\A(G-v)$.
\end{prop}
\begin{proof}
	Let $r$ be the number of neighbors of $v$ in $G$. We have $S(G,k)=(k-r)S(G-v,k)+S(G-v,k-1)$. Assuming that $G$ is of order $n$, we have
	\begin{align*}
	\B(G)&=\sum_{k=0}^{n}S(G,k)=\sum_{k=0}^{n-1}kS(G-v,k)-r\sum_{k=0}^{n-1}S(G-v,k)+\sum_{k=0}^{n-1}S(G-v,k)\\
	&=\sum_{k=0}^{n-1}(k-r+1)S(G-v,k)\\
	\mbox{and }\;\\
	\T(G)&=\sum_{k=0}^{n}kS(G,k)=\sum_{k=0}^{n}(k^2-kr)S(G-v,k)+\sum_{k=0}^{n}kS(G-v,k-1)\\
	&=\sum_{k=0}^{n-1}(k^2-kr)S(G-v,k)+\sum_{k=0}^{n-1}(k+1)S(G-v,k)\\
	&=\sum_{k=0}^{n}(k^2-k(r-1)+1)S(G-v,k).
	\end{align*}
	\noindent We therefore have
	\begin{align*}
&\B(G-v)\T(G)-\T(G-v)\B(G) \\
=&\sum_{k=0}^{n-1}S(G-v,k)\sum_{k'=0}^{n}(k'^2{-}k'(r{-}1){+}1)S(G-v,k')-\sum_{k=0}^{n-1}kS(G-v,k)\sum_{k'=0}^{n-1}(k'{-}r{+}1)S(G-v,k')\\
=&\sum_{k=0}^{n-1}\left(S(G-v,k)\right)^2(k^2-k(r-1)+1-k(k-r+1)\\
&+\sum_{k=0}^{n-2}\sum_{k'=k+1}^{n-1}S(G-v,k)S(G-v,k')(k'^2{-}k'(r{-}1){+}1{+}k^2{-}k(r{-}1){+}1{-}k(k'{-}r{+}1){-}k'(k{-}r{+}1))\\
=&\sum_{k=0}^{n-1}\left(S(G-v,k)\right)^2+\sum_{k=0}^{n-2}\sum_{k'=k+1}^{n-1}S(G-v,k)S(G-v,k')\left((k-k')^2+2\right)>0
\end{align*}
which implies $\A(G)-\A(G-v)=\displaystyle\frac{\T(G)}{\B(G)}-\frac{\T(G-v)}{\B(G-v)}=\frac{\B(G-v)\T(G)-\T(G-v)\B(G)}{\B(G)\B(G-v)}~>~0.\qedhere$	
\end{proof}



\begin{prop}\label{prop:7}  Let $G, H$ and $F_1,\cdots,F_r$ be $r+2$ graphs, and let $\alpha_1,\cdots,\alpha_r$ be $r$ positive numbers such that
\begin{itemize}
	\item $\B(G)=\B(H)+\sum_{i=1}^r\alpha_i\B(F_i)$
	\item $\T(G)=\T(H)+\sum_{i=1}^r\alpha_i\T(F_i)$
	\item $\A(F_i)<\A(H)$ for all $i=1,\cdots,r$.
\end{itemize}
Then $\A(G)<\A(H)$.
\end{prop}
\begin{proof}
	Since $\A(F_i)<\A(H)$, we have $\T(F_i)<\frac{\T(H)\B(F_i)}{\B(H)}$ for $i=1,\cdots,r$. Hence,
	\begin{align*}
	\A(G)&=\frac{\T(G)}{\B(G)}=\frac{\T(H)+\sum_{i=1}^r\alpha_i\T(F_i)}{\B(H)+\sum_{i=1}^r\alpha_i\B(F_i)}\\
	&<\frac{\T(H)+\sum_{i=1}^r\alpha_i\frac{\T(H)\B(F_i)}{\B(H)}}{\B(H)+\sum_{i=1}^r\alpha_i\B(F_i)}=\frac{\T(H)\left(\B(H)+\sum_{i=1}^r\alpha_i\B(F_i)\right)}
	{\B(H)\left(\B(H)+\sum_{i=1}^r\alpha_i\B(F_i)\right)}\\
	&=\frac{\T(H)}{\B(H)}=\A(H).\qedhere
	\end{align*}
	\end{proof}

\section{Inequalities for the Bell numbers}\label{sec:3}
In this section, we show how to derive inequalities for the Bell numbers, using properties related to the average number $\A(G)$ of colors in non-equivalent colorings of $G$. We start by analyzing paths. As already mentioned, $P_n\cup pK_1$ is the graph obtained by adding $p$ isolated vertices to a path on $n$ vertices.

\begin{theorem}
	$\A(P_{n}\cup (p+1)K_1)<\A(P_{n+1}\cup pK_1)$ for all $n\geq 1$ and $p\geq 0$.
\end{theorem}

\begin{proof}
	It follows from Equations (\ref{rec_plus}) that 
	\begin{align*}
	&\B(P_{n}\cup (p+1)K_1)=\B(P_{n+1}\cup pK_1)+\B(P_{n}\cup pK_1)\\
	\mbox{and}\quad\quad&\\
&\T(P_{n}\cup (p+1)K_1)=\T(P_{n+1}\cup pK_1)+\T(P_{n}\cup pK_1).\end{align*}
	Also, we know from Proposition \ref{prop:6} that $\A(P_{n}\cup pK_1)<\A(P_{n+1}\cup pK_1)$. Hence, it follows from Proposition \ref{prop:7} that $\A(P_{n}\cup (p+1)K_1)<\A(P_{n+1}\cup pK_1)$.
\end{proof}
\noindent Proposition \ref{prop:3} immediately gives the following Corollary.
\begin{corollary}\label{cor:9}If $n\geq 1$ and $p\geq 0$ then 
	\begin{align*}
	\mbox{   }\quad\displaystyle\left(\sum_{i = 0}^{p+1} {p+1 \choose i} B_{n+i}\right)
	\left(\sum_{i = 0}^{p} {p \choose i} B_{n+i}\right)
	<
	\left(\sum_{i = 0}^{p+1} {p+1 \choose i} B_{n+i-1}\right)
	\left(\sum_{i = 0}^{p} {p \choose i} B_{n+i+1}\right).\end{align*}
\end{corollary}

\begin{examples}
	\noindent For $p=0$ and $n\geq 1$, Corollary \ref{cor:9} provides the following inequality  :
	$$(B_n+B_{n+1})B_n<(B_{n-1}+B_n)B_{n+1}
	\iff B_n^2<B_{n-1}B_{n+1}.
$$
	This inequality for the Bell numbers also follows from Proposition \ref{prop:6}. Indeed, $P_n$ is obtained from $P_{n+1}$ by removing a vertex of degree 1, which implies
	$$\A(P_n)<\A(P_{n+1})\iff \frac{B_{n}}{B_{n-1}}<\frac{B_{n+1}}{B_{n}}\iff B_n^2<B_{n-1}B_{n+1}.$$
		
		Note that  Engel~\cite{Engel94} has shown that the sequence $(B_n)_{n\geq 0}$ is log-convex, which implies ${B_n^2\leq B_{n-1}B_{n+1}}$ (with a non-strict inequality) for $n\geq 1$. Recently, Alzer~\cite{Alzer19} has proved that the sequence $(B_n)_{n\geq 0}$ is strictly log-convex by showing that  
	$$
	B_{n-1}B_{n+1} - B_n^2 = \frac{1}{2e^2} \sum_{k=2}^{\infty} \sum_{j=1}^{k-1} \frac{j^{n-1}(k-j)^{n-1}}{j!(k-j)!}(k-2j)^2
	$$
	for all $n\geq 2$. Since $B_1^2=1<2=B_{0}B_{2}$, this also implies  $B_n^2<B_{n-1}B_{n+1}$ for all $n\geq 1$.\\
	
	\noindent As a second example, assume $p=1$ and $n\geq 1$. Corollary \ref{cor:9} provides the following inequality for the Bell numbers, which also follows from the strict log-convexity of the sequence $(B_n)_{n\geq 0}$:
	\begin{align*}
	&(B_n+B_{n+1}+B_{n+2})(B_n+B_{n+1})<(B_{n-1}+B_n+B_{n+1})(B_{n+1}+B_{n+2})\\
	\iff&B_{n}(B_{n}+B_{n+1})<B_{n-1}(B_{n+1}+B_{n+2}).
	\end{align*}
\end{examples}
For $n\geq 3$ and $r\geq 0$, we denote $H_{n,r}$ the graph obtained by linking one extremity of $P_r$ to one vertex of $C_n$ (see Figure \ref{fig:HNP}). For $r=0$,  $H_{n,0}$ is equal to $C_n$.
Also, $H_{n,r}\cup pK_1$ is the graph obtained from $H_{n,r}$ by adding $p$ isolated vertices. We now compare $\A(H_{3,n-3}\cup pK_1)$ with $\A(P_{n+1}\cup pK_1)$ to derive new inequalities involving the Bell numbers.

\begin{figure}[!hbtp]
	\centering
	\includegraphics[scale = 1.1]{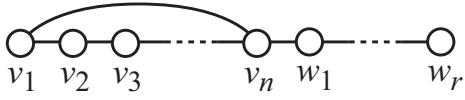}
	\caption{The graph $H_{n,r}$}\label{fig:HNP}
\end{figure}

\begin{theorem}\label{thm:10}
	$\A(H_{3,n-3}\cup pK_1)<\A(P_{n+1}\cup pK_1)$ for all $n\geq 4$ and $p\geq 0$.
\end{theorem}
\begin{proof}
Note first that Equations (\ref{rec_minus}) give $\B(H_{3,n-3}\cup pK_1)=\B(P_n\cup pK_1)-\B(P_{n-1}\cup pK_1)$ and $\T(H_{3,n-3}\cup pK_1)=\T(P_n\cup pK_1)-\T(P_{n-1}\cup pK_1)$.
Hence,
\begin{align*}
&\A(P_{n+1}\cup pK_1)-\A(H_{3,n-3}\cup pK_1)=
\frac{\T(P_{n+1}\cup pK_1)}{\B(P_{n+1}\cup pK_1)}-\frac{\T(H_{3,n-3}\cup pK_1)}{\B(H_{3,n-3}\cup pK_1)}\\
=&\frac{\T(P_{n+1}\cup pK_1)}{\B(P_{n+1}\cup pK_1)}-
\frac{\T(P_n\cup pK_1)-\T(P_{n-1}\cup pK_1)}{\B(P_n\cup pK_1)-\B(P_{n-1}\cup pK_1)}\\
=&\frac{\T(P_{n+1}\cup pK_1)\left(\B(P_n\cup pK_1)-\B(P_{n-1}\cup pK_1)\right)
	-\B(P_{n+1}\cup pK_1)\left(\T(P_n\cup pK_1)-\T(P_{n-1}\cup pK_1)\right)}
{\B(P_{n+1}\cup pK_1)\B(H_{3,n-3}\cup pK_1)}.
\end{align*}
Let $f(n,p)$ be the numerator of the above fraction. It follows from Proposition \ref{prop:3} that
$$f(n,p)=\sum_{i = 0}^{p} {p \choose i} B_{n+i+1}\left(\sum_{\ell = 0}^{p} {p \choose \ell} (B_{n+\ell-1}-B_{n+\ell-2})\right)
-\sum_{i = 0}^{p} {p \choose i} B_{n+i}\left(\sum_{\ell = 0}^{p} {p \choose \ell} (B_{n+\ell}-B_{n+\ell-1})\right).$$ 
It remains to prove that $f(n,p)>0$ for all $n\geq 4$ and $p>0$.
Since $B_n=\frac{1}{e}\sum_{k=1}^{\infty}\frac{k^{n}}{k!}$, we have 
$$B_{n}-B_{n-1}=\frac{1}{e}\sum_{k=1}^{\infty}\left(\frac{k^{n}}{k!}-\frac{k^{n-1}}{k!}\right)=\frac{1}{e}\sum_{k=2}^{\infty}\frac{k^{n-1}}{k!}(k-1).$$
Hence,
\begin{align*}
e^2f(n,p)=&\sum_{i=0}^{p}\sum_{\ell=0}^{p}\sum_{j=1}^{\infty}\sum_{k=2}^{\infty}{p \choose i}{p \choose \ell}\frac{j^{n+i+1}}{j!}\frac{k^{n+\ell-2}}{k!}(k-1)
-
\sum_{i=0}^{p}\sum_{\ell=0}^{p}\sum_{j=1}^{\infty}\sum_{k=2}^{\infty}{p \choose i}{p \choose \ell}\frac{j^{n+i}}{j!}\frac{k^{n+\ell-1}}{k!}(k-1)\\
=&\sum_{i=0}^{p}\sum_{\ell=0}^{p}\sum_{j=1}^{\infty}\sum_{k=2}^{\infty}{p \choose i}{p \choose \ell}\frac{j^{n+i}}{j!}\frac{k^{n+\ell-2}}{k!}(k-1)(j-k)\\
=&\sum_{i=0}^{p}\sum_{\ell=0}^{p}\sum_{j>k\geq 1}{p \choose i}{p \choose \ell}\left(\frac{j^{n+i}k^{n+\ell-2}}{j!k!}(k-1)(j-k)-\frac{j^{n+i-2}k^{n+\ell}}{j!k!}(j-1)(j-k)\right)\\
=&\sum_{i=0}^{p}\sum_{\ell=0}^{p}\sum_{j>k\geq 1}{p \choose i}{p \choose \ell}\frac{j^{n-2}k^{n-2}}{j!k!}(j-k)
\left(j^{i+2}k^{\ell}(k-1)-j^{i}k^{\ell+2}(j-1)\right)\\
=&\sum_{j>k\geq 1}\frac{j^{n-2}k^{n-2}}{j!k!}(j-k)\left(
(k-1)j^2\sum_{i=0}^{p}{p \choose i}j^i\sum_{\ell=0}^{p}{p \choose \ell}k^{\ell}
-(j-1)k^2\sum_{i=0}^{p}{p \choose i}j^i\sum_{\ell=0}^{p}{p \choose \ell}k^{\ell}
\right)\\
=&\sum_{j>k\geq 1}\frac{j^{n-2}k^{n-2}}{j!k!}(j-k)
(j+1)^p(k+1)^p\left((k-1)j^2-(j-1)k^2\right).
\end{align*}
Let $g(j,k)=(k-1)j^2-(j-1)k^2$. We have proved that 
$$e^2f(n,p)=\sum_{j>k\geq 1}\frac{j^{n-2}k^{n-2}}{j!k!}(j-k)
(j+1)^p(k+1)^pg(j,k).$$
Note that $g(j,1)=1-j$, $g(j,2)=(j-2)^2$, and
$$g(j,k)=j^2k-j^2-jk^2+k^2=(j-k)(jk-j-k)=(j-k)\left(k(\frac{j}{2}-1)+j(\frac{k}{2}-1)\right).$$
Hence, $g(j,k)>0$ for $j>k\geq 3$, and it remains to prove that
$$\sum_{k=1}^2\sum_{j=k+1}^{\infty}\frac{j^{n-2}k^{n-2}}{j!k!}(j-k)
(j+1)^p(k+1)^pg(j,k)>0.$$
We have
\begin{align*}
&\sum_{k=1}^2\sum_{j=k+1}^{\infty}\frac{j^{n-2}k^{n-2}}{j!k!}(j-k)
(j+1)^p(k+1)^pg(j,k)\\
=&\sum_{j=3}^{\infty}2^{n-3}\frac{j^{n-2}}{j!}3^p(j+1)^p(j-2)^3-
\sum_{j=2}^{\infty}\frac{j^{n-2}}{j!}2^p(j+1)^p(j-1)^2\\
\geq&12^{p}\left(\sum_{j=3}^{5}2^{n-3}\frac{j^{n-2}}{j!}(j-2)^3-
\sum_{j=2}^{5}\frac{j^{n-2}}{j!}(j-1)^2\right)+\sum_{j=6}^{\infty}\frac{j^{n-2}}{j!}2^p(j+1)^p\left(
2^{n-3}(j-2)^3-(j-1)^2
\right)\\
=&12^{p}\left(\frac{1}{6}2^{n-3}3^{n-2}+\frac{8}{24}2^{n-3}4^{n-2}+\frac{27}{120}2^{n-3}5^{n-2}-\frac{1}{2}2^{n-2}-\frac{4}{6}3^{n-2}-\frac{9}{24}4^{n-2}-\frac{16}{120}5^{n-2}\right)\\
&+\sum_{j=6}^{\infty}\frac{j^{n-2}}{j!}2^p(j+1)^p\left(
2^{n-3}(j-2)^3-(j-1)^2\right).
\end{align*}
It is easy to check that $\frac{1}{6}2^{n-3}3^{n-2}+\frac{8}{24}2^{n-3}4^{n-2}+\frac{27}{120}2^{n-3}5^{n-2}-\frac{1}{2}2^{n-2}-\frac{4}{6}3^{n-2}-\frac{9}{24}4^{n-2}-\frac{16}{120}5^{n-2}{>}0$ for all $n\geq 4$, and that $2^{n-3}(j-2)^3-(j-1)^2>0$ for all $n\geq 4$ and $j\geq 4$. Hence $f(n,p)>0$.
\end{proof}

As shown in the above proof, the above Theorem is equivalent to the following inequalities for the Bell numbers.

\begin{corollary}If $n\geq 4$ and $p\geq 0$ then
	\begin{align*}
	\sum_{i = 0}^{p} {p \choose i} B_{n+i+1}\left(\sum_{\ell = 0}^{p} {p \choose \ell} (B_{n+\ell-1}-B_{n+\ell-2})\right)
	>\sum_{i = 0}^{p} {p \choose i} B_{n+i}\left(\sum_{\ell = 0}^{p} {p \choose \ell} (B_{n+\ell}-B_{n+\ell-1})\right).\end{align*}
\end{corollary}

\begin{example}

For $p=0$ and $n\geq 4$  we get the following inequality for the Bell numbers:

$$B_{n+1}(B_{n-1}-B_{n-2})>B_n(B_{n}-B_{n-1}).$$
\end{example}

We now compare the average number of colors in colorings of $C_{n}\cup pK_1$ with the average number of colors in colorings of $H_{3,n-3}\cup pK_1$.

\begin{lemma}\label{lem:16} If $n\geq 3$, $r\geq 0$ and $p\geq 0$ then\\
	
	$\B(H_{n,r}\cup pK_1)=\begin{cases}
	\displaystyle\sum_{i=0}^{\frac{n-3}{2}}{\frac{n-3}{2} \choose i}\B(H_{3,2i+r}\cup pK_1)&\mbox{if }n\mbox{ is odd}\\
	\displaystyle\sum_{i=0}^{\frac{n-4}{2}}{\frac{n-4}{2} \choose i}\B(H_{3,2i+r+1}\cup pK_1)+\B(P_{2+r}\cup pK_1)&\mbox{if }n\mbox{ is even}
	\end{cases}$
	
	$\T(H_{n,r}\cup pK_1)=\begin{cases}
	\displaystyle\sum_{i=0}^{\frac{n-3}{2}}{\frac{n-3}{2} \choose i}\T(H_{3,2i+r}\cup pK_1)&\mbox{if }n\mbox{ is odd}\\
	\displaystyle\sum_{i=0}^{\frac{n-4}{2}}{\frac{n-4}{2} \choose i}\T(H_{3,2i+r+1}\cup pK_1)+\T(P_{2+r}\cup pK_1)&\mbox{if }n\mbox{ is even.}
	\end{cases}$
	
\end{lemma}

\begin{proof}
	The result is clearly valid for $n=3$. 
	For $n=4$, Equations~(\ref{rec_minus}) and~(\ref{rec_plus}) give
	\begin{align*}
	\B(H_{4,r}\cup pK_1)&=\B(P_{4+r}\cup pK_1)-\B(H_{3,r}\cup pK_1)\\
	&=\left(\B(H_{3,r+1}\cup pK_1)+\B(P_{3+r}\cup pK_1)\right)-
	\left(\B(P_{3+r}\cup pK_1)-\B(P_{2+r}\cup pK_1)\right)\\
	&=\B(H_{3,r+1}\cup pK_1)+\B(P_{2+r}\cup pK_1).
	\end{align*}
	Similarly, $\T(H_{4,r}\cup pK_1){=}\T(H_{3,r+1}\cup pK_1){+}\T(P_{2+r}\cup pK_1)$ which shows that the result is valid for $n=4$.  For larger values of $n$, we proceed by induction. Hence, it is sufficient to prove that $\B(H_{n,r}\cup pK_1)=\B(H_{3,n-3+r}\cup pK_1)+\B(H_{n-2,r}\cup pK_1)$ and $\T(H_{n,r}\cup pK_1){=}\T(H_{3,n-3+r}\cup pK_1){+}\T(H_{n-2,r}\cup pK_1)$. Using Equations~(\ref{rec_minus}) and~(\ref{rec_plus}), we get
	\begin{align*}
	\B(H_{n,r}\cup pK_1)=&\B(P_{n+r}\cup pK_1)-\B(H_{n-1,r}\cup pK_1)\\
	=&\left(\B(H_{3,n-3+r}{\cup} pK_1){+}\B(P_{n+r-1}{\cup} pK_1)\right){-}\left(\B(P_{n+r-1}{\cup} pK_1){-}\B(H_{n-2,r}{\cup} pK_1)\right)\\
	=&\B(H_{3,n-3+r}\cup pK_1)+\B(H_{n-2,r}\cup pK_1).
	\end{align*}
	The proof for $\T(H_{n,r}\cup pK_1)$ is similar.
\end{proof}

\noindent We are now ready to compare $\A(C_{n}\cup pK_1)$ with $\A(H_{3,n-3}\cup pK_1)$.

\begin{theorem}\label{thm:18} $\A(C_{n}\cup pK_1)<\A(H_{3,n-3}\cup pK_1)$ for all $n\geq 3$ and $p\geq 0$.
\end{theorem}
\begin{proof}
Lemma \ref{lem:16} (with $r=0$) gives  :

$\B(C_n\cup pK_1)=\begin{cases}
\displaystyle\B(H_{3,n-3}\cup pK_1)+\sum_{i=0}^{\frac{n-5}{2}}\B(H_{3,2i}\cup pK_1)&\mbox{if }n\mbox{ is odd}\\
\displaystyle\B(H_{3,n-3}\cup pK_1)+\sum_{i=0}^{\frac{n-6}{2}}\B(H_{3,2i+1}\cup pK_1)+\B(P_2\cup pK_1)&\mbox{if }n\mbox{ is even}
\end{cases}$

\noindent and

$\T(C_n\cup pK_1)=\begin{cases}
\displaystyle\T(H_{3,n-3}\cup pK_1)+\sum_{i=0}^{\frac{n-5}{2}}\T(H_{3,2i}\cup pK_1)&\mbox{if }n\mbox{ is odd}\\
\displaystyle\T(H_{3,n-3}\cup pK_1)+\sum_{i=0}^{\frac{n-6}{2}}\T(H_{3,2i+1}\cup pK_1)+\T(P_2\cup pK_1)&\mbox{if }n\mbox{ is even.}
\end{cases}$
	
We know from Proposition \ref{prop:6} that $\A(P_2\cup pK_1)<\A(C_3\cup pK_1)=\A(H_{3,0}\cup pK_1)$. Since $H_{3,n-3}\cup pK_1$ is obtained from $H_{3,i}\cup pK_1$ ($i<n-3$) by repeatedly adding vertices of degree 1, we have $\A(P_2\cup pK_1)<\A(H_{3,i}\cup pK_1)<\A(H_{3,n-3}\cup pK_1)$ for all $i=0,\cdots,n-5$. Proposition \ref{prop:7} therefore implies $\A(C_{n}\cup pK_1)<\A(H_{3,n-3}\cup pK_1)$.
\end{proof}

Equations (\ref{rec_minus}) give $\B(H_{3,n-3}\cup pK_1)=\B(P_n\cup pK_1)-\B(P_{n-1}\cup pK_1)$ and $\T(H_{3,n-3}\cup pK_1)=\T(P_n\cup pK_1)-\T(P_{n-1}\cup pK_1)$. Hence, Propositions \ref{prop:3} and \ref{prop:5} immediately give the following Corollary.
\begin{corollary}If $n\geq 3$ and $p\geq 0$ then
	\begin{align*}
	&\left(\sum_{j=1}^{n-1}(-1)^{j+1}\sum_{i = 0}^{p} {p \choose i} B_{n+i-j+1}\right)
	\left(
	\sum_{i = 0}^{p} {p \choose i} (B_{n+i-1}-B_{n+i-2})
	\right)\\
	<&\left(
	\sum_{j=1}^{n-1}(-1)^{j+1}\sum_{i = 0}^{p} {p \choose i} B_{n+i-j}
	\right)
	\left(
	\sum_{i = 0}^{p} {p \choose i} (B_{n+i}-B_{n+i-1})
	\right).
	\end{align*}
\end{corollary}

\begin{example}

\noindent For $p=0$ and $n\geq 3$, the above Corollary provides the following inequalities for the Bell numbers:
$$\displaystyle (B_{n-1}-B_{n-2})\sum_{j=1}^{n-1}(-1)^{j+1}B_{n-j+1}
<
(B_n-B_{n-1})\sum_{j=1}^{n-1}(-1)^{j+1}B_{n-j}.$$
It is easy to check that this inequality is also valid for $n=2$.\\
\end{example}
We now compare the average number of colors in colorings of paths with the average number of colors in colorings of cycles.
\begin{theorem}\label{thm:12}
	$\A(C_n\cup pK_1)>\A(P_n\cup pK_1)$ for all $n\geq 5$ and $p\geq 0$.
\end{theorem}
\begin{proof}
	We know from Theorems \ref{thm:10} and \ref{thm:18} that $\A(P_n{\cup} pK_1)>\A(H_{3,n-4}{\cup} pK_1)>\A(C_{n-1}{\cup} pK_1)$, which implies: $$\T(C_{n-1}{\cup} pK_1)<\frac{\T(P_n\cup pK_1)\B(C_{n-1}{\cup}pK_1)}{\B(P_n{\cup} pK_1)}.$$ Equations (\ref{rec_minus}) show that $\B(C_n\cup pK_1)=\B(P_n\cup pK_1)-\B(C_{n-1}\cup pK_1)$ and $\T(C_n\cup pK_1)=\T(P_n\cup pK_1)-\T(C_{n-1}\cup pK_1)$. Hence :	
		\begin{align*}
		\A(C_n\cup pK_1)&=\frac{\T(C_n\cup pK_1)}{\B(C_n\cup pK_1)}=\frac{\T(P_n\cup pK_1)-\T(C_{n-1}\cup pK_1)}{\B(P_n\cup pK_1)-\B(C_{n-1}\cup pK_1)}\\&>\frac{\T(P_n\cup pK_1)-\frac{\T(P_n\cup pK_1)\B(C_{n-1}\cup pK_1)}{\B(P_n\cup pK_1)}}{\B(P_n\cup pK_1)-\B(C_{n-1}\cup pK_1)}\\
		&=\frac{\T(P_n\cup pK_1)\left(\B(P_n\cup pK_1)-\B(C_{n-1}\cup pK_1)\right)}
		{\B(P_n\cup pK_1)\left(\B(P_n\cup pK_1)-\B(C_{n-1}\cup pK_1)\right)}
	=\frac{\T(P_n\cup pK_1)}{\B(P_n\cup pK_1)}=\A(P_n\cup pK_1).
		\end{align*}
%
\end{proof}

\noindent Propositions \ref{prop:3} and \ref{prop:5} immediately give the following Corollary.
\begin{corollary}If $n\geq 5$ and $p\geq 0$ then
	\begin{align*}
	&\left(\sum_{i = 0}^{p} {p \choose i} B_{n+i-1}\right)\left(\sum_{j=1}^{n-1}(-1)^{j+1}\sum_{i = 0}^{p} {p \choose i} B_{n+i-j+1}\right)\\
	>&
	\left(\sum_{i = 0}^{p} {p \choose i} B_{n+i}\right)\left(\sum_{j=1}^{n-1}(-1)^{j+1}\sum_{i = 0}^{p} {p \choose i} B_{n+i-j}\right).\end{align*}
\end{corollary}

\begin{example}
	
	\noindent For $p=0$ and $n\geq 5$ the above Corollary provides the following inequality for the Bell numbers:
	$$\displaystyle B_{n-1}\sum_{j=1}^{n-1}(-1)^{j+1}B_{n-j+1}
	>
	B_n\sum_{j=1}^{n-1}(-1)^{j+1}B_{n-j}.$$
	
\end{example}

\noindent We finally compare $\A(C_{n}\cup pK_1)$ with $\A(C_{n-2}\cup (p+2)K_1)$.

\begin{lemma}\label{lem:15} If $n\geq 3$ and $p\geq 0$ then
	\begin{align*}
	\B(C_n\cup (p+2)K_1)=\B(H_{n,2}\cup pK_1)+2\B(H_{n,1}\cup pK_1)+\B(H_{n,0}\cup pK_1)\\
	\T(C_n\cup (p+2)K_1)=\T(H_{n,2}\cup pK_1)+2\T(H_{n,1}\cup pK_1)+\T(H_{n,0}\cup pK_1).
	\end{align*}
\end{lemma}
\begin{proof}
	Equations~\eqref{rec_plus} give
	\begin{align*}
	\B(C_n\cup (p+2)K_1)&=\B(H_{n,1}\cup (p+1)K_1)+\B(C_n\cup (p+1)K_1)\\
	&=\left(\B(H_{n,2}\cup pK_1)+\B(H_{n,1}\cup pK_1)\right)+
	\left(\B(H_{n,1}\cup pK_1)+\B(H_{n,0}\cup pK_1)\right)\\
	&=\B(H_{n,2}\cup pK_1)+2\B(H_{n,1}\cup pK_1)+\B(H_{n,0}\cup pK_1).
	\end{align*}
	The proof is similar for $\T(C_n\cup (p+2)K_1)$.
\end{proof}

\begin{theorem} $\A(C_{n}\cup pK_1)>\A(C_{n-2}\cup (p+2)K_1)$ for all $n\geq 5$ and $p\geq 0$.
	\end{theorem}
	\begin{proof}
We divide the proof into two cases, according to the parity of $n$.\\

\noindent {\it Case 1 : $n$ is odd}

\noindent Lemma \ref{lem:16} (with $r=0$) shows that
 \allowdisplaybreaks
 
$$\displaystyle\B(C_n\cup pK_1)=\sum_{i=0}^{\frac{n-3}{2}}\B(H_{3,2i}\cup pK_1)\; \mbox{ and } \displaystyle\T(C_n\cup pK_1)=\sum_{i=0}^{\frac{n-3}{2}}\T(H_{3,2i}\cup pK_1)$$
and Lemmas \ref{lem:15} and \ref{lem:16} give
\begin{align*}
&\B(C_{n-2}\cup (p+2)K_1)\\
=&B(H_{n-2,2}\cup pK_1)+2\B(H_{n-2,1}\cup pK_1)+\B(H_{n-2,0}\cup pK_1)\\
=&\sum_{i=0}^{\frac{n-5}{2}}{\frac{n-5}{2} \choose i}\B(H_{3,2i+2}\cup pK_1)
+2\sum_{i=0}^{\frac{n-5}{2}}{\frac{n-5}{2} \choose i}\B(H_{3,2i+1}\cup pK_1)
+\sum_{i=0}^{\frac{n-5}{2}}{\frac{n-5}{2} \choose i}\B(H_{3,2i}\cup pK_1)\\
=&\sum_{i=1}^{\frac{n-3}{2}}{\frac{n-5}{2} \choose i-1}\B(H_{3,2i}\cup pK_1)
+2\sum_{i=0}^{\frac{n-5}{2}}{\frac{n-5}{2} \choose i}\B(H_{3,2i+1}\cup pK_1)
+\sum_{i=0}^{\frac{n-5}{2}}{\frac{n-5}{2} \choose i}\B(H_{3,2i}\cup pK_1)\\
=&\sum_{i=1}^{\frac{n-5}{2}}\left({\frac{n-5}{2} \choose i-1}+{\frac{n-5}{2} \choose i}\right)\B(H_{3,2i}\cup pK_1)+\B(H_{3,n-3}\cup pK_1) + \B(H_{3,0}\cup pK_1)\\
&+2\sum_{i=0}^{\frac{n-5}{2}}{\frac{n-5}{2} \choose i}\B(H_{3,2i+1}\cup pK_1)\\
=&\sum_{i=0}^{\frac{n-3}{2}}{\frac{n-3}{2} \choose i}\B(H_{3,2i}\cup pK_1)
+2\sum_{i=0}^{\frac{n-5}{2}}{\frac{n-5}{2} \choose i}\B(H_{3,2i+1}\cup pK_1)\\
=&\B(C_n\cup pK_1)+\sum_{i=1}^{n-4}\alpha_i\B(H_{3,i}\cup pK_1)
\end{align*} 
where 
$$\alpha_i=\begin{cases}
\;\;\displaystyle{\frac{n-3}{2} \choose \frac{i}{2}}-1&\mbox{if }i\mbox{ is even}\\[0.5cm]
\displaystyle 2{\frac{n-5}{2} \choose \frac{i-1}{2}}&\mbox{if }i\mbox{ is odd.}
\end{cases}$$
Similarly, $\T(C_{n-2}\cup (p+2)K_1)=\T(C_n\cup pK_1)+\sum_{i=1}^{n-4}\alpha_i\T(H_{3,i}\cup pK_1)$.
Moreover, we know from Theorems \ref{thm:10} and \ref{thm:12} that 
$$\A(H_{3,n-4}\cup pK_1)<\A(P_{n}\cup pK_1)<\A(C_{n}\cup pK_1).$$ Also, given $i\in \{1,\cdots,n-5\}$, $H_{3,n-4}\cup pK_1$ is obtained from $H_{3,i}\cup pK_1$ by repeatedly adding vertices of degree 1, and it follows from Proposition \ref{prop:6} that 
$$\A(H_{3,i}\cup pK_1)<\A(H_{3,n-4}\cup pK_1)<\A(C_{n}\cup pK_1).$$ 
Since all $\alpha_i$ are strictly positive, we can conclude from Proposition \ref{prop:7} that $\A(C_{n-2}\cup (p+2)K_1)<\A(C_{n}\cup (p+2)K_1)$.\\

\noindent {\it Case 2 : $n$ is even}

\noindent The proof is similar to the previous case. More precisely, Lemma \ref{lem:16} shows that

$\displaystyle\B(C_n\cup pK_1)=\sum_{i=0}^{\frac{n-4}{2}}\B(H_{3,2i+1}\cup pK_1)+\B(P_{2}\cup pK_1)$ and

 $\displaystyle\T(C_n\cup pK_1)=\sum_{i=0}^{\frac{n-4}{2}}\T(H_{3,2i+1}\cup pK_1)+\T(P_{2}\cup pK_1)$
 
 \noindent and lemmas \ref{lem:15} and \ref{lem:16} give
 \allowdisplaybreaks
 \begin{align*}
 &\B(C_{n-2}\cup (p+2)K_1)=B(H_{n-2,2}\cup pK_1)+2\B(H_{n-2,1}\cup pK_1)+\B(H_{n-2,0}\cup pK_1)\\
 =&\left(\sum_{i=0}^{\frac{n-6}{2}}{\frac{n-6}{2} \choose i}\B(H_{3,2i+3}{\cup} pK_1){+}\B(P_{4}{\cup} pK_1)\right){+}2\left(\sum_{i=0}^{\frac{n-6}{2}}{\frac{n-6}{2} \choose i}\B(H_{3,2i+2}{\cup} pK_1){+}\B(P_{3}{\cup} pK_1)\right)\\
 &+\left(\sum_{i=0}^{\frac{n-6}{2}}{\frac{n-6}{2} \choose i}\B(H_{3,2i+1}\cup pK_1)+\B(P_{2}\cup pK_1)\right)\\
 =&\sum_{i=1}^{\frac{n-4}{2}}{\frac{n-6}{2} \choose i-1}\B(H_{3,2i+1}\cup pK_1)+2\sum_{i=0}^{\frac{n-6}{2}}{\frac{n-6}{2} \choose i}\B(H_{3,2i+2}\cup pK_1)+\sum_{i=0}^{\frac{n-6}{2}}{\frac{n-6}{2} \choose i}\B(H_{3,2i+1}\cup pK_1)\\
 &+\B(P_{4}\cup pK_1)+2\B(P_{3}\cup pK_1)+\B(P_{2}\cup pK_1)\\
 =&\sum_{i=1}^{\frac{n-6}{2}}\left({\frac{n-6}{2} \choose i-1}+{\frac{n-6}{2} \choose i}\right)\B(H_{3,2i+1}\cup pK_1)
 +\B(H_{3,n-3}\cup pK_1) + \B(H_{3,1}\cup pK_1)\\
 &+2\sum_{i=0}^{\frac{n-6}{2}}{\frac{n-6}{2} \choose i}\B(H_{3,2i+2}\cup pK_1)+\B(P_{4}\cup pK_1)+2\B(P_{3}\cup pK_1)+\B(P_{2}\cup pK_1)\\
 =&\sum_{i=0}^{\frac{n-4}{2}}{\frac{n-4}{2} \choose i}\B(H_{3,2i+1}\cup pK_1)
 +2\sum_{i=0}^{\frac{n-6}{2}}{\frac{n-6}{2} \choose i}\B(H_{3,2i+2}\cup pK_1)\\
 &+\B(P_{4}\cup pK_1)+2\B(P_{3}\cup pK_1)+\B(P_{2}\cup pK_1)\\
 =&\B(C_n\cup pK_1)+\sum_{i=2}^{n-4}\alpha_i\B(H_{3,i}\cup pK_1)+\B(P_{4}\cup pK_1)+2\B(P_{3}\cup pK_1)
 \end{align*} 
 where 
 $$\alpha_i=\begin{cases}
 \;\;\displaystyle{\frac{n-4}{2} \choose \frac{i-1}{2}}-1&\mbox{if }i\mbox{ is odd}\\[0.5cm]
 \displaystyle 2{\frac{n-6}{2} \choose \frac{i-2}{2}}&\mbox{if }i\mbox{ is even.}
 \end{cases}$$
 Similarly, 
 $$\displaystyle\T(C_{n-2}\cup (p+2)K_1)=\T(C_n\cup pK_1)+\sum_{i=2}^{n-4}\alpha_i\T(H_{3,i}\cup pK_1)+\T(P_{4}\cup pK_1)+2\T(P_{3}\cup pK_1).$$
 As already mentioned, we know that 
 $$\A(H_{3,i}\cup pK_1)<\A(H_{3,n-4}\cup pK_1)<\A(P_{n}\cup pK_1)<\A(C_{n}\cup pK_1)$$ for all $i=2,\cdots,n-4.$
 Also, $P_{n}\cup pK_1$ is obtained from $P_{3}\cup pK_1$ by repeatedly adding vertices of degree 1, and it follows from Proposition \ref{prop:6} that 
 $$\A(P_{3}\cup pK_1)<\A(P_{4}\cup pK_1)<\A(P_{n}\cup pK_1)<\A(C_{n}\cup pK_1).$$  Since all $\alpha_i$ are strictly positive, we conclude from Proposition \ref{prop:7} that $\A(C_{n-2}\cup (p+2)K_1)<\A(C_{n}\cup (p+2)K_1)$.
	\end{proof}
	
\noindent Proposition \ref{prop:5} immediately gives the following Corollary.

\begin{corollary}If $n\geq 5$ and $p\geq 0$ then
	\begin{align*}
	&\left(\sum_{j=1}^{n-3}(-1)^{j+1}\sum_{i = 0}^{p+2} {p+2 \choose i} B_{n+i-j-1}\right)\left(\sum_{j=1}^{n-1}(-1)^{j+1}\sum_{i = 0}^{p} {p \choose i} B_{n+i-j}\right)\\
	<&\left(
	\sum_{j=1}^{n-3}(-1)^{j+1}\sum_{i = 0}^{p+2} {p+2 \choose i} B_{n+i-j-2}\right)\left(\sum_{j=1}^{n-1}(-1)^{j+1}\sum_{i = 0}^{p} {p \choose i} B_{n+i-j+1}\right).
	\end{align*}
\end{corollary}
\begin{example}

\noindent For $p=0$ and $n\geq 5$, the above Corollary provides the following inequalities for the Bell numbers:
\begin{align*}
&\sum_{j=1}^{n-3}(-1)^{j+1}(B_{n-j-1}+2B_{n-j}+B_{n-j+1})\sum_{j=1}^{n-1}(-1)^{j+1}B_{n-j}\\
<&
\sum_{j=1}^{n-3}(-1)^{j+1} (B_{n-j-2}+2B_{n-j-1}+B_{n-j})\sum_{j=1}^{n-1}(-1)^{j+1} B_{n-j+1}\\
\iff&
\left(\sum_{j=2}^{n-2}(-1)^{j}B_{n-j}+2\sum_{j=1}^{n-3}(-1)^{j+1}B_{n-j}+\sum_{j=0}^{n-4}(-1)^{j}B_{n-j}\right)\sum_{j=1}^{n-1}(-1)^{j+1}B_{n-j}\\
<&\left(\sum_{j=2}^{n-2}(-1)^{j} B_{n-j-1}+2\sum_{j=1}^{n-3}(-1)^{j+1} B_{n-j-1}+\sum_{j=0}^{n-4}(-1)^{j} B_{n-j-1}\right)\sum_{j=1}^{n-1}(-1)^{j+1} B_{n-j+1}\\
\iff&
(B_n+B_{n-1}+7(-1)^n)\sum_{j=1}^{n-1}(-1)^{j+1}B_{n-j}<(B_{n-1}+B_{n-2}+3(-1)^n)\sum_{j=1}^{n-1}(-1)^{j+1} B_{n-j+1}.
\end{align*}

It is easy to check that this inequality is also valid for $n=4$.
\end{example}
\section{Conclusion}

We have shown how the average number of colors in the non-equivalent colorings of a graph $G$ helps to derive inequalities for the Bell numbers. Among the inequalities, we have shown that 
\begin{itemize}
	\item $B_n^2<B_{n-1}B_{n+1}\quad$ for all $n\geq 1$,
	\item $B_{n}(B_{n}+B_{n+1})<B_{n-1}(B_{n+1}+B_{n+2})\quad$ for all $n\geq 1$,
	\item $B_n(B_{n}-B_{n-1})< B_{n+1}(B_{n-1}-B_{n-2})\quad$ for all $n\geq 4$,
	\item $\displaystyle (B_{n-1}-B_{n-2})\sum_{j=1}^{n-1}(-1)^{j+1}B_{n-j+1}
	<
	(B_n-B_{n-1})\sum_{j=1}^{n-1}(-1)^{j+1}B_{n-j}\quad$ for all $n\geq 2$,
	\item $\displaystyle B_n\sum_{j=1}^{n-1}(-1)^{j+1}B_{n-j} < B_{n-1}\sum_{j=1}^{n-1}(-1)^{j+1}B_{n-j+1}\quad$ for all $n\geq 5$,
	\item $\displaystyle(B_n+B_{n-1}+7(-1)^n)\sum_{j=1}^{n-1}(-1)^{j+1}B_{n-j}<(B_{n-1}+B_{n-2}+3(-1)^n)\sum_{j=1}^{n-1}(-1)^{j+1} B_{n-j+1}\quad$ \\
	for all $n\geq 4$.
\end{itemize}
We have no doubt that other inequalities for the Bell numbers can be generated by comparing the average numbers of colors in the non-equivalent colorings of other types of graphs.

\noindent 2010 {\it Mathematics Subject Classification}: Primary
11B73; Secondary 05C15.

\noindent \emph{Keywords: } Bell number, Stirling number, graph coloring.

\bigskip
\hrule
\bigskip

\noindent (Concerned with sequences \seqnum{A141390},
\seqnum{A005493}.)

\end{document}